\documentclass{article}
\usepackage{amsmath,amsthm,amssymb}

\allowdisplaybreaks[3]

\numberwithin{equation}{section}

\newtheorem{theorem}{Theorem}[section]

\newtheorem{proposition}[theorem]{Proposition}
\theoremstyle{remark}
\newtheorem{remark}[theorem]{Remark}

\newcommand\R{{\mathbb R}}
\newcommand\D{{\mathbb D}}
\newcommand\X{{\R^d}}
\newcommand\N{{\mathbb N}}

\newcommand\al{\alpha}
\newcommand\la{\lambda}
\newcommand\Ga{\Gamma}
\newcommand\ga{\gamma}
\newcommand\ka{\varkappa}

\newcommand\lu{\left\langle}
\newcommand\ru{\right\rangle}
\newcommand\llu{\lu\!\lu}
\newcommand\rru{\ru\!\ru}

\title{Fractional contact model  in the continuum
\thanks{The financial support of DFG through the SFB 701
(Bielefeld University) is gratefully acknowledged.}}

\author{Anatoly N. Kochubei\thanks{Institute of Mathematics,
National Academy of Sciences of Ukraine, Kyiv, Ukraine ({\tt
kochubei@imath.kiev.ua}).} \and Yuri G. Kondratiev\thanks{Fakult\"{a}t
f\"{u}r Mathematik, Universit\"{a}t Bielefeld, 33615 Bielefeld,
Germany ({\tt kondrat@math.uni-bielefeld.de})} }

\begin{document}

\maketitle
\begin{abstract}
We consider the evolution of correlation functions in a non-Markov version of the contact
model in the continuum. The memory effects are introduced by assuming the fractional evolution equation for the statistical dynamics. This leads to a behavior of time-dependent correlation functions, essentially different from the one known for the standard contact model.
\end{abstract}

\section{Statistical dynamics of continuous systems}

A traditional interpretation of a dynamics is given in terms of trajectories in a phase space of the considered system.
This may be formulated as an evolution of microscopic states of the system. However such microscopic dynamics creates naturally an associated evolution of functions on the phase space (observables) and
the dual evolution of probability measures on the phase space (macroscopic states). This is a classical
triple of evolution equations in the Hamiltonian systems theory where the Newton description of
particles dynamics may be reformulated by means of Liouville equations for observables or states.

If we are dealing with  Markov stochastic processes, then we have to use the language of random trajectories
for the stochastic dynamics in the phase space and there are two related evolution equations. Namely,
we have corresponding
dynamics of functions  given by the  Kolmogorov backward equation and
related dynamics of measures  given by the Kolmogorov forward or Fokker-Planck equation. In all cases,
we have evolution equations for probability measures on phase spaces.  These measures, or macroscopic states, describe statistical properties of considered systems and we will call corresponding evolution the statistical dynamics associated with the initial (random) dynamical system. This concept has sense even in the case when dynamics in terms of trajectories can not be constructed due a complicated character of considered models. We illustrate this situations in the framework of Markov evolutions of interacting particle systems in the continuum, see e.g. \cite{FKK12}.

The phase space of our systems  is the configuration space
\[
\Gamma :=\left\{ \left. \gamma \subset {\R}^{d}\right| \,|\gamma
\cap \Lambda |<\infty \,\,\mathrm{for\,\,any\,\,compact}\,\,\Lambda
\subset {\R}^{d}\right\} ,
\]
where $|A|$ denotes the cardinality of a set $A\subset \X$. This space has natural interpretation as a space
of counting measures on $\X$  and is endowed with the vague topology induced from the space of Radon measures on $\X$.

Let $L$ be a Markov pre-generator defined on some set of functions $\mathcal{F}(\Ga)$ given
on the configuration space
$\Ga$.
The backward Kolmogorov  equation on observables reads as:
\begin{equation}
\label{BKE}
\frac{\partial F_{t}}{\partial t}=\,LF_{t},
\end{equation}
$$
F_{t}|_{t=0}=F_{0}\in \mathcal{F}(\Ga).
$$
The duality between observables and states on $\Ga$ is given by

  $$\llu F,\,\mu\rru :=\int_{\Ga}Fd\mu . $$
Then the forward Kolmogorov, or   Fokker- Planck, equation is the dual one for (\ref{BKE}):
\begin{equation}
\label{FKE}
\frac{\partial \mu_{t}}{\partial
t}=\,L^{\ast}\mu_{t},
\end{equation}
$$
\mu_{t}|_{t=0}=\mu_{0}.
$$
To be able to construct a Markov process for the generator $L$ we need to define  transition probabilities
$P_t (\ga, d\ga')$ for this process which correspond to the solution to (\ref{FKE}) for the initial Dirac measures, i.e.,
$\mu_0(d\ga') = \delta_{\ga}(d\ga'),\;\;\ga \in \Ga$. There appears a principal difference comparing with the classical theory of Markov processes. Namely, we can hope to have the existence of measure evolutions in
(\ref{FKE}) only for special classes of initial states. The latter means that we may have a statistical dynamics for the considered system but a possibility to construct the stochastic evolution remains an open
problem.

Moreover, even in such a restricted approach an analysis of Fokker-Planck equations for measures on $\Ga$
is a difficult task and at the present time direct methods for the study of such kind  evolution equations
are absent.  There exists a useful technical possibility to transfer evolution problems for measures to their
characteristics such as correlation functions.

Let $f^{(n)}(x_1,\dotsc, x_n)$ be a compactly supported continuous symmetric function  on $\left({\R}^{d}\right)^n$  and $\mu$ be a  given probability measure on $\Ga$.
The correlation function $k_\mu^{(n)}$ of order $n\in\N$ for the measure $\mu$ is defined by the following
relation
\begin{multline*}
\int_\Gamma \sum_{\{x_{i_1},\dotsc,x_{i_n}\}\subset \gamma}
f^{(n)}(x_{i_1},\dotsc,x_{i_n})\,d\mu(\gamma) \\
=\frac{1}{n!} \int_{({\R}^d)^n} f^{(n)}(y_1,\dotsc,y_n)
k_\mu^{(n)}(y_1,\dotsc,y_n) \,dy_1\dotsm dy_n ,
\end{multline*}
and we put $k_\mu^{(0)}=1$.
Of course, the existence of correlation functions  needs certain a priori properties of the measure
$\mu$. On the other hand, correlation functions reconstruct  a measure under proper condition of their
positive definiteness and a bound on  growth, see  \cite{KK02}.

The Fokker-Planck equation (\ref{FKE}) in terms of time dependent correlation functions may be rewritten as an infinite system of evolution equations (hierarchical chain)

$$ \frac{\partial k_{t}^{(n)}}{\partial t} = (L^{\triangle} k_t)^{(n)}, n\geq 0 .$$
Here
 $L^{\triangle}$ is the image of $L^{*}$  in a Fock-type space
 of vector-functions $k_t= \left( k_{t}^{(n)}\right )_{n=0}^{\infty} $.
In applications to concrete models, the expression for operator $L^{\triangle}$ is  typically  easy to obtain from $L$
using  combinatoric calculations, c.f. \cite{FKO09}.

The approach to the construction
and study of statistical Markov dynamics described above was successfully applied in a number of interacting particle models, see, e.g.,  \cite{FKK12}. But there exists a possibility to apply similar approach to the case of non-Markov evolutions. Namely, in the Fokker-Planck evolution equation (\ref{FKE}) we may use the
Caputo--Djrbashian (CD)
fractional time derivative $\D^{\al}_t$  instead  of  the usual time derivative.
For functions $f:\R_+ \to \R$ the CD fractional derivative provides a fractional generalization of the first derivative through the following formula in the Laplace transform domain
$$
\left( \mathcal{L} \D^{\al}_t f\right)(s)= s^{\al} (\mathcal{L} f)(s) - s^{\al -1} f(0),\;\; s>0, \al\in(0,1].
$$
where
$$
(\mathcal{L} f)(s) =\int_0^{\infty}
e^{-st} f(t)dt.
$$
Another representation of the CD derivative is
$$
\D^{\al}_t f(t)= \frac{1}{\Ga (1-\al)} \frac{d}{dt} \int_0^t \frac{f(\tau)-f(0)}{(t-\tau)^\al } d\tau,\; 0<\al<1.
$$
If $f$ is absolutely continuous, then
$$
\D^{\al}_t f(t)= \frac{1}{\Ga (1-\al)}  \int_0^t \frac{f'(\tau)}{(t-\tau)^\al } d\tau,\; 0<\al<1.
$$
In this form the CD fractional derivative is often used in physical literature.

The definition above have natural extensions for vector-valued or measure-valued functions on $\R_+$.
 In the case $0<\al<1$ we shall expect that
the  fractional Fokker-Planck dynamics (if exists)  will act in the space of states on $\Ga$, i.e., will preserve probability measures on $\Ga$. Of course, this property shall be stated rigorously for each particular model under consideration.
From the analytic point of view, the fractional Fokker-Planck equation
\begin{equation}
\label{FFKE}
\D^{\al}_t \mu_t =\,L^{\ast}\mu_{t},
\end{equation}
$$
\mu_{t}|_{t=0}=\mu_{0}.
$$
describes a dynamical system with memory  in the space of measures on $\Ga$. The corresponding evolution has no longer the semigroup property but we still can try to construct the solution $\mu_t$ by means of a related
hierarchical chain:
$$
\D^{\al}_t k_{t}^{(n)} = (L^{\triangle} k_t)^{(n)}, n\geq 0 .
$$
In the present paper we illustrate this approach in the case of so-called contact model in the continuum.
This model was introduced in \cite{KS06} as a Markov birth-and-death process in $\Ga$ and studied in detail from the point of view of statistical dynamics in
\cite{KKP08}.
The Markov generator for this process on a proper set of observables $F:\Ga\to \R$ is given by
$$
(LF)(\ga)= \sum_{x\in\ga} [F(\ga \setminus \{x\})-F(\ga)] +
\ka \sum_{x\in \ga} \int_{\X} a(x-y) [F(\ga \cup \{y\}) -F(\ga)] dy,
$$
where $0\leq a\in L^{1}(\X)$ is even and $\varkappa >0$.
Roughly speaking, the contact process describes branching of points of a configuration in $\X$
in the presence of given constant mortality rate for these points which is independent on the existing
configuration. We will see that the use of fractional time derivative will not change the equilibrium states of the model but will lead to essentially different asymptotic for the time dependent correlation functions.

\section{Fractional Cauchy problem}
Let us consider the following Cauchy problem
\begin{equation}
\D^{(\alpha)} u = Au +f,\;\; t > 0,\;\; u(0)=x,
\label{CP}
\end{equation}
where $\D^{(\alpha)}$ is the CD fractional derivative, $0 <\alpha <1$, $A$  is a generator of a $C_0$-semigroup
$T(t)$ on a Banach
space $X$.
If $x\in X$,  and the vector-function
$$
t\ni t \mapsto \int_0^t f(s) ds
$$
is exponentially bounded, then a mild solution of (\ref{CP})  is given \cite{B00}, \cite{KLW13} by
\begin{equation}
u(t) = S_\alpha (t)x +\int_0^t P_\alpha (t-s) f(s) ds ,
\end{equation}
where
\begin{equation}
S_\alpha (t) x = \int_{0}^{\infty} \Phi_\alpha(\tau) T(\tau t^\alpha) x d\tau, \;\;t\geq 0;
\end{equation}
$$
P_\alpha (t) y = \alpha t^{\alpha - 1} \int_{0}^{\infty} \tau \Phi_\alpha (\tau) T(\tau t^\alpha) y d\tau, \;\;t\geq 0,\; x,y \in X.
$$

Here $\Phi_\alpha(z)$ is the Wright function:
$$
\Phi_\alpha(z) = \sum_{n=0}^\infty \frac{(-z)^n}{n! \Ga (-\alpha n +1 -\alpha)}.
$$
It is known (see \cite{GLM99}) that
$$
\Phi_\alpha(t) \geq 0,\;\;t >0;\;\;\; \int_0^\infty \Phi_\alpha(t)dt =1,
$$
$$
\int_0^\infty  t \Phi_\alpha(t) dt =\frac{1}{\alpha \Ga(\alpha)}.
$$
If $T$ is a contraction semigroup, then
$$
\|S_\alpha(t) x\| \leq \|x\|, \;\; \|P_\alpha (t)y\| \leq \frac{t^{\alpha - 1}}{\Ga (\alpha)} \|y\|.
$$
More generally, we have \cite{GLM99}
$$
\int_0^\infty  t^n  \Phi_\alpha(t) dt =\frac{n!}{ \Ga(1 +\alpha n)},\;\;n=0,1,2, \dots  .
$$
If
$$
\|T(t)\| \leq e^{\lambda t},
$$
then
$$
\|S_\al (t)\| \leq \int_0^\infty \Phi_\al (\tau) e^{\la t^\al \tau}d\tau =
\sum_{n=0}^\infty  \frac{ (\la t^\al )^n }{\Ga (1+\al n)} = E_\al (\la t^\al);
$$
$$
\|P_\al (t)\| \leq \al t ^{\al -1} \int_0^\infty \tau \Phi_\al (\tau) e^{\la t^\al \tau} d\tau
= \al t^{\al-1} \sum_{n=0}^\infty \frac{(\la t^\al)^n }{n!} \int_0^\infty \tau^{n+1} \Phi_\al (\tau) d\tau
$$
$$
= t^{\al -1} \sum_{n=0}^\infty \frac{(\la t^\al)^n}{\Ga(\al n +\al)}= t^{\al -1} E_{\al,\al}(\la t^\al).
$$
Here $E_{\al}$ and $E_{\al,\al}$ are the Mittag-Leffler function and generalized Mittag-Leffler function respectively:
$$
E_\al (z)= \sum_{n=0}^\infty  \frac{ z^n }{\Ga (1+\al n)},
$$
$$
E_{\al,\al} (z)= \sum_{n=0}^\infty  \frac{ z^n }{\Ga (\al+\al n)}.
$$
Here we follow notations from \cite{KLW13}, \cite{KST06}.

\section{Fractional contact model}

Let us consider the hierarchical chain of evolution equations for correlation functions of the contact model in which we use
the CD derivative in time instead of usual time derivative as in \cite{KKP08}. It leads to the following Cauchy problem:
\begin{equation}
\label{FracCP}
\D_t^{(\al)}k_{t}^{(n)}(x_{1},\ldots,x_{n})
={\widehat{L}_{n}^{\ast}}k_{t}^{(n)}(x_{1},\ldots,x_{n})+f_{t}^{(n)}(x_{1},\ldots,x_{n}),\;\;\;n\geq
1,
\end{equation}
\begin{equation}
\label{FracCP0}
k_{t}^{(n)}(x_{1},\ldots,x_{n})|_{t=0} = k_{0}^{(n)}(x_{1},\ldots,x_{n}),
\end{equation}
$$
0\leq k_{0}^{(n)}(x_{1},\ldots,x_{n})\leq C^n n!,\quad C\ge 1.
$$
Here
$$
{\widehat{L}_{n}^{\ast}}k_{t}^{(n)}(x_{1},\ldots,x_{n}):=-n\,k_{t}^{(n)}(x_{1},\ldots,x_{n})+
$$
$$
+\varkappa\sum_{i=1}^{n}\int_{\R^{d}}a(x_{i}-y)k_{t}^{(n)}(x_{1},\ldots,x_{i-1},y,x_{i+1},\ldots,x_{n})dy,\;\;\;n\geq
1
$$
and
$$
f_{t}^{(n)}(x_{1},\ldots,x_{n}):=\varkappa\sum_{i=1}^{n}k_{t}^{(n-1)}(x_{1},\ldots,\check{x_{i}},\ldots,x_{n})\sum_{j:\,j\neq
i}a(x_{i}-x_{j}),\;\;\;n\geq 2,
$$
$$
f_{t}^{(1)}\equiv 0.
$$
We assume that $0\leq a\in L^{1}(\X) \cap L^\infty(\X)$ is even and $\varkappa >0$.

Let $n\in\N$ be arbitrary and fixed. We consider the linear Cauchy
problem
\begin{equation}
 \D_t^{(\al)} k_{t}^{(n)}(x_{1},\ldots,x_{n})={\widehat{L}_{n}^{\ast}}k_{t}^{(n)}(x_{1},\ldots,x_{n})+f_{t}^{(n)}(x_{1},\ldots,x_{n}),\;\;t\geq
0,\label{Cauchy}
\end{equation}
$$
\left.k_{t}^{(n)}(x_{1},\ldots,x_{n})\right|_{t=0}:=k_{0}^{(n)}(x_{1},\ldots,x_{n}),
$$
in the Banach space $\mathrm{X}_{n}$ of bounded measurable functions on
$(\R^{d})^n$.

The operator ${\widehat{L}_{n}^{\ast}}$ in $\mathrm{X}_{n}$
can be written also in another way
$$
{\widehat{L}_{n}^{\ast}}k^{(n)}(x_{1},\ldots,x_{n})=n(\varkappa-1)\,k^{(n)}(x_{1},\ldots,x_{n})+
\sum_{i=1}^{n}L_{a}^{i}k^{(n)}(x_{1},\ldots,x_{n}),
$$

where for each $\,1 \leq i \leq n,$
$$
L_{a}^{i}k^{(n)}(x_{1},\ldots,x_{n})=
$$
$$
=\varkappa\int_{\R^{d}}a(x_{i}-y)
\left[k^{(n)}(x_{1},\ldots,x_{i-1},y,x_{i+1},\ldots,x_{n})-k^{(n)}(x_{1},\ldots,x_{n})\right]dy
$$
is a generator of a Markov process on $(\R^{d})^n$ which describes the jump of the particle placed at
the point $(x_{1},\ldots,x_{i},\ldots,x_{n})\in(\R^{d})^n$ into the
point $(x_{1},\ldots,y,\ldots,x_{n})\in(\R^{d})^n$ with intensity
equal to $\,a(x_{i}-y)$.

Because $a\in {L}^{1}(\R^{d})$,
for any $n\geq 1$ the operator ${\widehat{L}_{n}^{\ast}}$ is a
bounded linear operator on
$\mathrm{X}_{n}$ (as well as on ${L}^{1}((\R^{d})^{n})$).
Moreover, for each $\,1 \leq i \leq n,$ the operator $L_{a}^{i}$  acting in variable $x_i$ is
a generator of a contraction semigroup on the space of bounded measurable functions (and on $L^1(\X)$); see \cite{KKP08}.
Then $\bigotimes_{i=1}^{n}e^{tL_{a}^{i}}$ is a contraction semigroup on 
$\mathrm{X}_{n}$ (and  on ${L}^{1}((\R^{d})^{n})$).

Let
$$
T^{(n)}(t) =e^{n(\varkappa - 1)t} \bigotimes_{i=1}^{n}e^{tL_{a}^{i}}.
$$
We have $\|T^{(n)}(t)\| \leq e^{n(\varkappa - 1)t}$. By (2.2), the solution to  (\ref{FracCP}), (\ref{FracCP0})
is
$$
k_t^{(n)} (x_1,\dots, x_n)= (S_{\al}^{(n)}(t)k_0^{(n)})(x_1,\dots,x_n)+
$$
$$
\int_0^t P_\al^{(n)}(t-s)\left[ \varkappa \sum_{i=1}^{n}  k_{s}^{(n-1)} ( \dots,  \check{x_i}, \dots) \sum_{j\neq i} a(\cdot_{i} -\cdot_{j})\right]
ds
$$
(for $n=1$ the second summand is absent).
Here $S_\al^{(n)}$ and $P_\al^{(n)}$ correspond as above to the semigroup $T^{(n)}$.

\section{Correlation functions bounds}

Our estimates will be based on the identity (\cite{D93}, page 2)
\begin{equation}
\label{Dj}
\int_0^t (t-\tau)^{\al -1} E_{\al,\al}(z(t-\tau)^\al) E_\al(\la\tau^\al) d\tau= \frac{E_\al(z t^\al) - E_\al(\la t^\al) }{z-\la}.
\end{equation}

Denote $A=\max \{ 1,\sup\limits_{x\in \R^d}a(x)\}$.

\begin{proposition}
If $\varkappa > 1$, then
\begin{multline}
\label{>1}
k_t^{(n)}(x_1,\dots,x_n) \\
\leq  E_\al (n(\ka -1)t^\al)\left[ C^n n!  + C^{n-1}  n! \sum_{j=0}^{n-2} \left(\frac{\ka A}{\ka - 1}\right)^{j+1}
 (n-1)\dots(n-j -1)\right] .
 \end{multline}
The second term is absent for $n=1$.

\end{proposition}

\begin{proof}

 We show (\ref{>1}) by induction using the fact that due to (\ref{Dj})
$$
\int_0^t (t-\tau)^{\al -1} E_{\al,\al}((n+1)(\ka -1)(t-\tau)^\al) E_\al(n(\ka -1)\tau^\al) d\tau=
$$
$$
\frac{E_\al((n+1)(\ka -1)t^\al) - E_\al(
n(\ka -1)t^\al) }{\ka -1} \leq \frac{1}{\ka -1} E_\al ((n+1)(\ka-1)t^\al).
$$

Indeed, the inequality (\ref{>1}) is obvious for $n=1$. Suppose it has been proved for some value of $n$. Then it follows from the estimates for $S_\alpha^{(n)}$ and $P_\alpha^{(n)}$ that

\begin{multline*}
k_t^{(n+1)}(x_1,\dots,x_{n+1}) \leq C^{n+1}(n+1)!E_\al ((n+1)(\ka-1)t^\al)\\
+C^n(n+1)!\ka An\int\limits_0^t (t-\tau)^{\alpha -1}E_{\al,\al}((n+1)(\ka -1)(t-\tau)^\al)E_\al(n(\ka -1)\tau^\al) d\tau \\
+C^n(n+1)!\ka An\sum\limits_{j=0}^{n-2}\left( \frac{\ka A}{\ka -1}\right)^{j+1}(n-1)(n-2)\cdots (n-j-1)\\
\times \int\limits_0^t (t-\tau)^{\alpha -1}E_{\al,\al}((n+1)(\ka -1)(t-\tau)^\al)E_\al(n(\ka -1)\tau^\al) d\tau \\
\le E_\al((n+1)(\ka -1)t^\al)\Biggl[ C^{n+1}(n+1)! +C^n(n+1)!\frac{\ka An}{\ka -1} \\
+C^n(n+1)!\sum\limits_{j=0}^{n-2}\left( \frac{\ka A}{\ka -1}\right)^{j+2}n(n-1)(n-2)\cdots (n-j-1)\Biggr] .
\end{multline*}

Making the change $i=j+1$ of the summation index we see that
\begin{multline*}
\sum\limits_{j=0}^{n-2}\left( \frac{\ka A}{\ka -1}\right)^{j+2}n(n-1)(n-2)\cdots (n-j-1)\\
=\sum\limits_{i=1}^{(n+1)-2}\left( \frac{\ka A}{\ka -1}\right)^{i+1}[(n+1)-1]\cdots [(n+1)-i-1]
\end{multline*}
while the missing summand corresponding to $i=0$ coincides with the expression $ \dfrac{\ka An}{\ka -1}$ available in the above estimate. This results in (\ref{>1}) with $n+1$ substituted for $n$, thus in the estimate (\ref{>1}) for any $n$.

\end{proof}


It is known that \cite{D93},  \cite{KST06}
\begin{equation}
\label{MLAS}
E_\al(s) \sim \frac{1}{\al}e^{s^{\frac{1}{\al}}},\;\;\;s\to \infty.
\end{equation}
By (\ref{>1}) and (\ref{MLAS}),
\begin{equation}
\label{SUPER}
k_t^{(n)}(x_1,\dots,x_n) \leq   M\; C^n n!(n-1)!\frac{q^n}{q-1} e^{[n(\ka -1)]^{\frac{1}{\al}}t},
\end{equation}
$$
q=\frac{\ka A}{\ka -1}, \;\;\; M>0.
$$

\begin{remark}
The case $\varkappa > 1$ corresponds to the super-critical regime in the contact model.
As in the case of Markov contact model, we may expect that  bound (\ref{SUPER}) is exact, i.e., a similar
kind of the lower estimate  will be valid; see \cite{FKK09}.  For a given $n\in \N$  this bound means
an exponential growth of correlation functions for $t\to\infty$ as in the Markov case.   But for a given $t>0$
we observe much stronger growth of correlations w.r.t. $n\in\N$.

\end{remark}

\begin{proposition}
If $\ka <1$, then
\begin{multline}
\label{<1}
k_t^{(n)}(x_1,\dots,x_n) \leq   C^n n! E_\al (-n(1 - \ka) t^\al)+
C^{n-1} n! \\ \times \sum_{j=0}^{n-2} \left(\frac{\ka A}{1-\ka}\right)^{j+1} \frac{(n-1)(n-2)\dots (n-j-1)}{(j+1)!}
E_\al(-(n-j-1)(1-\ka) t^\al).
\end{multline}
The second term is absent for $n=1$.

\end{proposition}

\begin{proof}
We proceed by induction, using the fact that by (\ref{Dj}) for $m<n$,

$$
\int_0^t (t-\tau)^{\al -1} E_{\al,\al}(-n(1-\ka)(t-\tau)^\al) E_\al(-m(1-\ka)\tau^\al) d\tau=
$$
$$
\frac{E_\al(-m(1-\ka) t^\al) - E_\al(
-n(1-\ka)t^\al) }{(n-m)(1-\ka)} \leq \frac{1}{(n-m)(1-\ka)} E_\al (-m(1-\ka)t^\al)
$$
(note that the function $s\mapsto E_\al(-s)$ is monotonically decreasing, see \cite{D93}, page 6).

The inequality (\ref{<1}) is obvious for $n=1$. Suppose it has been proved for some value of $n$. Then
\begin{multline*}
k_t^{(n+1)}(x_1,\dots,x_{n+1}) \leq   C^{n+1} (n+1)! E_\al (-(n+1)(1 - \ka) t^\al)\\
+C^n(n+1)!\ka An\int_0^t (t-\tau)^{\al -1} E_{\al,\al}(-(n+1)(1-\ka)(t-\tau)^\al) E_\al(-n(1-\ka)\tau^\al) d\tau \\
+C^{n-1}(n+1)!\ka An\sum_{j=0}^{n-2} \left(\frac{\ka A}{1-\ka}\right)^{j+1} \frac{(n-1)(n-2)\dots (n-j-1)}{(j+1)!}\\
\times \int_0^t (t-\tau)^{\al -1} E_{\al,\al}(-(n+1)(1-\ka)(t-\tau)^\al) E_\al(-(n-j-1)(1-\ka)\tau^\al) d\tau \\
\le C^{n+1} (n+1)! E_\al (-(n+1)(1 - \ka) t^\al)+C^n (n+1)! \frac{\ka An}{1-\ka}E_\al (-n(1 - \ka) t^\al)\\
+C^n (n+1)! \frac{\ka An}{1-\ka}\sum_{j=0}^{n-2}\Biggl[ \left(\frac{\ka A}{1-\ka}\right)^{j+1} \frac{(n-1)(n-2)\dots (n-j-1)}{(j+2)!} \\
\times E_\al(-(n-j-1)(1-\ka) t^\al)\Biggr].
\end{multline*}

As in the proof of Proposition 4.1, we make the change of index $i=j+1$ and note that the "missing'' summand corresponding to $i=0$ is available in our estimate of $k_t^{(n+1)}$. This proves (\ref{<1}) for any $n$.

\end{proof}

It is known that \cite{D93},  \cite{KST06}
\begin{equation}
\label{Power}
E_\al (-s)\sim \frac1{\Gamma (1-\alpha)}s^{-1},\quad s\to +\infty .
\end{equation}
By (\ref{<1}) and (\ref{Power}), for $t\ge 1$,
\begin{equation}
\label{SUB}
k_t^{(n)}(x_1,\dots,x_n) \leq   MC^n n!(n-1)!\left( \frac{\ka A}{1-\ka }\right)^nt^{-\al},\quad M>0.
\end{equation}

\begin{remark}

For $\varkappa < 1$ we have the sub-critical regime
and in the Markov model correlation functions have exponential decay to zero for $t\to\infty$ \cite{KKP08}.
For $\al\in (0,1)$
and for given $n\in \N$ the bound (\ref{SUB}) is quite different: we have only $t^{-\al}$ decay that
demonstrates an essential effect of the memory on the rate of asymptotic degeneration of the system.

\end{remark}
\begin{proposition}
If $\ka =1$, then

\begin{multline}
\label{=1}
k_t^{(n)}(x_1,\dots,x_n) \leq   C^n n! \\
+ \al^{-1} C^{n-1} n! \sum_{j=0}^{n-2}
\frac{A^{j+1}}{(j+1)\Ga((j+1)\al)} (n-1)(n-2)\dots (n-j-1) t^{(j+1)\al}.
\end{multline}
The second term is absent for $n=1$.
\begin{proof}

As in the previous propositions, we proceed by induction, this time based on the identity
$$
\int_0^t (t-\tau)^{\al-1}\tau^{\beta -1} d\tau =
\frac{\Ga(\al)\Ga(\beta)}{\Ga(\al + \beta)} t^{\al+\beta -1}.
$$

Assuming (\ref{=1}) we find that
\begin{multline*}
k_t^{(n+1)}(x_1,\dots,x_{n+1}) \leq   C^{n+1} (n+1)! +C^n(n+1)!\frac{An}{\Ga (\al )}\int\limits_0^t (t-\tau )^{\al -1}d\tau \\
+C^{n-1}(n+1)!\frac{An}\alpha \sum_{j=0}^{n-2}
\frac{A^{j+1}}{(j+1)\Ga((j+1)\al)} (n-1)(n-2)\dots (n-j-1)\\
\times \frac1{\Ga (\al )}\int_0^t (t-\tau)^{\al-1}\tau^{(j+1)\al }d\tau .
\end{multline*}
Computing the integrals, using the identity $\Ga (1+z)=z\Ga (z)$ and making a change of the summation index, we obtain the required estimate of $k_t^{(n+1)}$.
\end{proof}
\end{proposition}

\end{document}